\documentclass{article}

\usepackage{subcaption}
\usepackage{graphicx}
\usepackage{enumerate}
\usepackage{natbib} 
\usepackage{url} 
\usepackage{xcolor}
\bibpunct[, ]{(}{)}{;}{a}{}{,}

\usepackage{hyperref}
\usepackage{mathtools}
\usepackage{amsthm}
\usepackage{amsfonts}
\usepackage{amssymb}
\usepackage{amstext}
\usepackage{bbm}
\usepackage{bbold}
\usepackage{enumitem}
\usepackage{geometry}
\usepackage{hyperref}
\usepackage{float}
\usepackage{mathrsfs}
\usepackage{lipsum} 

\numberwithin{equation}{section}

\theoremstyle{plain} 
\newtheorem{theorem}{Theorem}[section]      
\newtheorem{assumption}{Assumption}         
\newtheorem{corollary}[theorem]{Corollary}  
\newtheorem{algorithm}{Algorithm}[section]  

\theoremstyle{definition}
\newtheorem{definition}{Definition}
\newtheorem{example}{Example}

\newcommand{\Prob}{\mathsf{P}}
\newcommand{\Expect}{\mathsf{E}}

\newcommand{\far}{\mathsf{FAR}}
\newcommand{\pdc}{\mathsf{PDC}}

\newcommand{\wadd}{\mathsf{WADD}}

\newcommand{\dd}{\textnormal{D}}
\DeclareMathOperator*{\esssup}{ess\,sup}

\newcommand{\mg}{\mathcal{G}}

\newcommand{\msi}{\mathscr{I}}

\newcommand{\kl}{\textnormal{D}_{\textnormal{KL}}}

\title{Robust Quickest Change Detection with Sampling Control}
\author{Yingze Hou$^a$, Hoda Bidkhori$^b$ and 
Taposh Banerjee$^a$\thanks{CONTACT: Taposh Banerjee. Email: taposh.banerjee@pitt.edu}
    \hspace{.2cm}\\
    $^a$ Industrial Engineering, University of Pittsburgh\\
    $^b$ Computational and Data Sciences, George Mason University}
\date{}

\begin{document}

\maketitle
\begin{abstract}
The problem of quickest detection of a change in the distribution of a sequence of random variables is studied. The objective is to detect the 
change with the minimum possible delay, subject to constraints on the rate of false alarms and the cost of observations used in the decision-making process. The post-change distribution of the data is known only within a distribution family. It is shown that if the post-change family has a distribution that is least favorable in a well-defined sense, then 
a computationally efficient algorithm can be designed that uses an on-off observation control strategy to save the cost of observations. In addition, the algorithm can detect the change robustly while avoiding unnecessary false alarms. It is shown that the algorithm is also asymptotically robust optimal as the rate of false alarms goes to zero for every fixed constraint on the cost of observations. The algorithm's effectiveness is validated on simulated data and real public health data. 
%
\end{abstract}

\noindent%

\vspace{-0.3cm}
\section{Introduction}
In the Quickest Change Detection (QCD) problem for independent and identically distributed (i.i.d.) data, a decision maker observes a sequence of random variables, $\{X_n\}$. Before a change point, denoted by $\nu$, these random variables are i.i.d. with a fixed density $f$. Following the change at $\nu$, the sequence ${X_n}$ becomes i.i.d. with a different density $g$. The objective of QCD is to detect the shift in distribution from $f$ to $g$, as quickly as possible, while avoiding false alarms. As a result, stochastic optimization problems are solved to find a QCD algorithm, where the goal is to minimize a metric on the average detection delay subject to a constraint on a metric of false alarms. 
The QCD problem has applications in event detection or anomaly detection problems in finance, sensor networks, statistical process control, and public health  (\cite{poor-hadj-qcd-book-2009, veer-bane-elsevierbook-2013, tart-niki-bass-2014})

In the Bayesian setting of the QCD problem, where the change point $\nu$ is treated as a random variable with a known prior distribution, the optimal solution is obtained by minimizing the average detection delay (averaged over the prior on the change point) subject to a constraint on the probability of false alarm. The optimal solution in this setting is the Shiryaev test where a change is declared the first time the \textit{a-posteriori} probability that the change has already occurred crossed a well-designed threshold (\cite{shir-siamtpa-1963, tart-veer-siamtpa-2005}). Some recent results on QCD in the Bayesian setting can be found in (\cite{veer-bane-elsevierbook-2013, tart-niki-bass-2014, tart-book-2019, bane-tit-2021, guo2023bayesian, naha2024bayesian, hou2024robust}).

Without prior knowledge of the distribution for the change point $\nu$, the change point is regarded as an unknown constant. In this non-Bayesian setting, the concept of conditional delay (conditioned on the change point) is introduced, with the average detection delay depending on the timing of the change. Consequently, a minimax approach is adopted, to minimize the worst-case delay across all possible change points. The QCD problem is examined under two different minimax frameworks in (\cite{poll-astat-1985}) and (\cite{lord-amstat-1971}).
An optimal solution in this setting, for both minimax formulations, is the Cumulative Sum (CUSUM) algorithm, where the alarm is raised when an accumulated version of the log-likelihood ratios crosses a threshold
(\cite{page-biometrica-1954, lord-amstat-1971, mous-astat-1986, lai-ieeetit-1998}). A comprehensive review of the literature, along with recent advances in the non-Bayesian setting, can be found in (\cite{veer-bane-elsevierbook-2013, tart-niki-bass-2014, tart-book-2019, liang2022quickest, brucks2023modeling}).

When the post-change distribution is unknown, optimal algorithms are generally developed using three main approaches: (1) generalized likelihood ratio (GLR) tests, in which we replace the unknown post-change parameter with its maximum likelihood estimate (\cite{lord-amstat-1971, lai-ieeetit-1998, tart-niki-bass-2014, lau2018binning, tart-book-2019}); (2) mixture-based tests, in which we assume a prior model for the post-change parameters and integrate the likelihood ratio over this prior (\cite{poll-astat-1987, lai-ieeetit-1998, tart-niki-bass-2014, tart-book-2019}); and (3) robust tests, in which we design optimal tests using a least favorable distribution (\cite{unni-etal-ieeeit-2011, oleyaeimotlagh2023quickest, hou2024robust}). Among these, only the robust approach yields test statistics that are computationally efficient and can be calculated recursively.

In several applications of QCD, acquiring data or observations for decision-making is costly. The cost could be associated with the labor required to conduct experiments and collect data, the loss of revenue while performing a particular test (e.g., destructive testing), or the cost of computation, energy, or battery to process observations. This issue is amplified when changes occur rarely. In traditional statistical process control, the cost of taking observation is minimized by performing sampling control (\cite{reynolds2004control}). In a series of papers (\cite{banerjee2012data, banerjee2013data, banerjee2015data}), the problem of QCD with sampling control has been studied. In these works, optimal algorithms were obtained to minimize the detection delay, subject to constraints on the rate of false alarms and the average number of observations used before the change point. 
An optimal QCD algorithm with sampling control in these papers has two thresholds. The higher threshold is used for stopping and the lower threshold for on-off observation control. One remarkable fact proven in these papers is that these two-threshold algorithms have the same asymptotic delay and false alarm performances as the classical QCD algorithms without sampling control. In addition, one can achieve any arbitrary, but fixed, constraint on the sampling control.


One key limitation in the works \cite{banerjee2012data, banerjee2013data} is that to achieve optimality, one needs to know the precise pre- and post-change densities. This issue was partially addressed in the follow-up work \cite{banerjee2015data} where it was assumed that the post-change law is unknown and a GLR or a mixture approach was taken for QCD with sampling control. As discussed earlier, one major limitation of a GLR-based algorithm is that, although it allows us to achieve uniform optimality over all possible post-change parameters, it is computationally hard to implement. 

In this paper, we take a robust approach to QCD with sampling control when the post-change distribution is unknown. We obtain a computationally efficient algorithm using which we can avoid false alarms and save the cost of observations. In addition, the algorithm is robust in the sense that, under certain stochastic boundedness assumptions, it can detect all possible post-change scenarios. We also formally prove the robust optimality of the algorithm for a robust problem formulation.


We summarize our contributions here:
\begin{enumerate}
    \item We propose a robust QCD algorithm capable of sampling control where the post-change distribution is unknown. We design the algorithm using the notion of least favorable law (LFL) to be made precise below. See Section~\ref{sec:Algorithm}. 
    \item We show that the proposed algorithm can consistently detect all post-change scenarios. We also provide design guidelines to satisfy constraints on false alarms and the cost of observations. See Section~\ref{sec:Optimality}. 
    \item We propose a robust QCD formulation taking into account sampling control and show that our algorithm is asymptotically optimal with
respect to the proposed problem formulation. See Section \ref{sec:ProblemFormulation} and Section~\ref{sec:Optimality}. 
    \item We demonstrate the effectiveness of the algorithm on both simulated and real data. In particular, when applied to detect the onset of a pandemic, the robust algorithm, which selectively skips observations, showed nearly comparable performance to the robust CUSUM test that does not skip observations (see Section \ref{sec:Numerical}).
\end{enumerate}


\section{Problem Formulation}
\label{sec:ProblemFormulation}
We observe a sequence of random variables $\{X_n\}$ with the following property: 
\begin{equation}\label{eq:iidQCD}
	X_n \sim
	\begin{cases}
		f, &\quad \forall n < \nu, \\
		g, &\quad \forall n \geq \nu.
	\end{cases}
\end{equation}
Thus, before a time $\nu$, the random variables have density $f$ (with distribution $F$) and after $\nu$, the random variables have density $g$ (with distribution $G$). The variables are independent conditioned on the time $\nu$. 
The density $g$, or equivalently the distribution $G$, is unknown. However, we assume that $G$ belongs to a known family of distributions $\mathcal{G}$:
$$
\mathcal{G} = \{G: \text{$G$ is a possible post-change distribution}\}. 
$$

We wish to detect the change in density from $f$ to $g$ as quickly as possible while avoiding false alarms. However, not all observations can be used for decision-making because there is a cost associated with collecting observations. To control the cost of observations, we choose the mechanism of on-off observation control. It will be shown later in the paper that on-off observation controls are sufficient to achieve strong optimality results. Let $M_n$ be the indicator random variable such that 
\begin{align*}
    M_n = 
    \begin{cases}
	1, &\quad  \textnormal{if $X_n$ is used for decision-making}, \\
	0, &\quad  \textnormal{otherwise}.
    \end{cases}
\end{align*}
The information available at time $n$ is denoted by
$$\mathscr{I}_n = \left[M_1, \dots, M_n, X_1^{(M_1)}, \dots, X_n^{(M_n)}\right],$$
where $X_i^{(M_i)} = X_i$ if $M_i = 1$, otherwise $X_i$ is absent from the information vector $\mathscr{I}_n$. Also, $\mathscr{I}_0$ is an empty set.
The control $M_{n+1}$ is a function of the information available at time $n$. Namely,
$$M_{n+1} = \phi_{n+1}(\mathscr{I}_n),$$
where $\phi_{n+1}$ is some possibly time-dependent function. Thus, the decision on whether to take the observation at time $n+1$ is taken based on information available at time $n$. 

When a change is detected, the observation process is stopped. This time of stopping $\tau$ is also a part of the decision-making process. We choose $\tau$ as a stopping time on the information sequence $\{\mathscr{I}_n\}$. This means that the indicator function $\mathbbm{1}_{\{\tau = n\}}$ will be a function of $\msi_n$. More formally, the event $\{\tau \leq n\}$ is part of the sigma-algebra generated by $\msi_n$. 
Our policy for QCD with sampling control is
\begin{equation*}
    \Psi = \{\tau, \phi_1, \dots, \phi_{\tau}\}.
\end{equation*}

To find an optimal policy, we now define the metrics on delay, false alarm, and the cost of observations.  
For delay, we consider the minimax formulation of (\cite{lord-amstat-1971}):
\begin{equation}
    \label{eq:WADD G}
    \wadd^{G}(\Psi) \coloneqq \sup_k \esssup \Expect_k^{G}[(\tau - k + 1)^+ | \mathscr{I}_{k-1}].
\end{equation}
Here, we use $\Prob_{n}^{G}$ to denote the law of the observation process $\{X_{n}\}$ when the change occurs at $\nu=n$ and the post-change law is given by $G$, and $\Expect_{n}^{G}$ for the corresponding expectation. When there is no change (change occurring at $\nu=\infty$), we use $\Prob_{\infty}^{G}=\Prob_{\infty}$ and $\Expect_{\infty}^{G}=\Expect_{\infty}$ to denote the correspond measure and expectation.
The essential supremum operation $\esssup X$ gives the smallest constant $C$ such that $\Prob(X \leq C)=1$. 
The supremum over $k \geq 1$ is taken, since the change point $\nu = k$ is treated as an unknown constant.
Therefore, we consider the worst-case delay over the change points and past realizations.
While it may appear overly pessimistic, $\wadd$ is the only delay metric studied in the literature for which strong optimality results have been developed in both minimax and robust settings (\cite{unni-etal-ieeeit-2011, hou2024robust, bane-hou-multinonstationary}).


For the false alarm metric, without prior statistical knowledge of the change point, we consider the mean-time to a false alarm (assuming no change happens) or its reciprocal, the false alarm rate (\cite{lord-amstat-1971}):
\begin{equation}
    \label{eq:FAR}
    \far(\Psi) \coloneqq \frac{1}{\Expect_{\infty}[\tau]}.
\end{equation}
For a metric on the cost of observations, we use the Pre-change Duty Cycle ($\pdc$) metric first introduced in (\cite{banerjee2013data}):
\begin{equation}
    \label{eq:PDC}
    \pdc(\Psi) \coloneqq \limsup_k \frac{1}{k}\Expect_{k}^G\left[\sum_{i = 1}^{k-1} M_i \bigg| \tau \geq k\right]
    = \limsup_k \frac{1}{k}\Expect_{\infty}\left[\sum_{i = 1}^{k-1} M_i \bigg| \tau \geq k\right].
\end{equation}
The $\pdc$ is the average of the fraction of observations used before the change point $\nu=k$, $\frac{1}{k}\sum_{i = 1}^{k-1} M_i$ when the change occurs at a far horizon $(k \to \infty)$, and there is no false alarm $(\tau \geq k)$. Clearly, $\pdc \leq 1$. 
The expectation $\Expect_{\infty}$ is used 
because the fraction $\frac{1}{k}\sum_{i = 1}^{k-1} M_i$ only depends on the data used before the change point $\nu=k$. 

The minimax formulation studied in \cite{banerjee2013data}, where the post-change law $g$ is assumed known, is the following:
\begin{equation}
\label{problem 2 G}
    \begin{aligned}
        \inf_{\Psi} \ & \wadd^{G}(\Psi)\\
        \text{subject to } \; & \far(\Psi) \leq \alpha,\\
         \text{and }   \;  & \pdc(\Psi) \leq \beta,
    \end{aligned}
\end{equation}
where $\alpha$ is a constraint on the false alarm rate and $\beta$ is a constraint on the $\pdc$. 
With $\beta = 1$, problem \eqref{problem 2 G} reduces to the minimax formulation of \cite{lord-amstat-1971}.

However, since the post-change law $G$ is not exactly known to the decision-maker, we consider a robust version of the problem:
\begin{equation}
\label{problem 2 G beta}
    \begin{aligned}
        \inf_{\Psi}  \sup_{G \in \mg} \ & \wadd^{G}(\Psi)\\
        \text{subject to } \; & \far(\Psi) \leq \alpha,\\
         \text{and }    & \pdc(\Psi) \leq \beta.
    \end{aligned}
\end{equation}
We say that a policy $\Psi  = \{\tau, \phi_1, \dots, \phi_{\tau}\}$ is robust optimal if it is a solution to the problem in \eqref{problem 2 G beta}. We say that a policy is asymptotically robust optimal if it is a solution to the problem in \eqref{problem 2 G beta}, as $\alpha \to 0$, for a fixed $\beta$. 

\section{Algorithm for Robust QCD with Adaptive Sampling Control}
\label{sec:Algorithm}
In this section, we propose a computationally efficient QCD algorithm with sampling control that can be implemented without precisely knowing the post-change density $g$. In the next section, we will show that the proposed algorithm is asymptotically robust optimal for the problem in \eqref{problem 2 G beta}. To define the algorithm, we need the following definitions.

\begin{definition}[Stochastic boundedness]\label{def Stochastic boundedness}
Consider random variable $Z_1$ with density $f_1$ and random variable $Z_2$ with density $f_2$. We say $f_{1}$ is stochastically bounded by $f_{2}$, and use the notation $f_2 \succ f_1$ or $Z_2 \succ Z_1$ if 
$$
\Prob(Z_2 \geq t) \geq \Prob(Z_1 \geq t), \quad \forall t \in \mathbb{R}.
$$
\end{definition}

\begin{definition}[Least Favorable Law (LFL)]\label{assumption LFL}
We say the family of post-change laws $\mathcal{G}$ is stochastically bounded by a law $\bar G \in \mathcal{G}$ if
\begin{align*}
    \mathcal{L}\left(\log\frac{\bar g(X)}{f(X)}, G\right) \succ \mathcal{L}\left(\log\frac{\bar g(X)}{f(X)}, \bar G\right), \quad \forall G \in \mathcal{G},
\end{align*}
where $\mathcal{L}(\phi(X), G)$
denotes the law or distribution of $\phi(X)$ when $X$ has distribution $G$. We refer to $\bar G$ or its density $\bar g$ as the least favorable law. 
 \end{definition}

\begin{definition}
The Kullback-Leibler divergence between two densities $g$ and $f$ is defined as
$$
\dd_{\textnormal{KL}}(g \; \| \; f) := \int g(x) \log \frac{g(x)}{f(x)} dx.
$$    
\end{definition}

 \clearpage
\begin{assumption}
In the rest of the paper, we will make the following assumptions:
\label{assumption1}
\begin{enumerate}
    \item An LFL $\bar g$ exists. 
    \item All the likelihood ratios appearing in the paper are continuous and have finite positive moments up to second order. 
\end{enumerate}
    
\end{assumption}

We now define our algorithm, called the Robust Data-Efficient Cumulative Sum (RDE-CUSUM) algorithm, using the LFL $\bar g$.   

\begin{algorithm}[RDE-CUSUM Policy or Algorithm: $\bar{\Psi}_{\text{RDC}}(A, \mu, h)$]
\label{alg: robust DE-CUSUM}
Fix $A > 0$, $\mu > 0$, and $h \geq 0$. The \textup{RDE-CUSUM} statistic $\{\bar D_n\}$ and the corresponding sampling control $\{\bar M_n\}$ are defined as follows:  
\begin{enumerate}
    \item Sampling control: 
    \begin{equation}
        \bar M_{n+1} = \begin{cases}
            1, \quad \text{if } \bar D_n \geq 0, \\
            0, \quad \text{if } \bar D_n < 0. 
        \end{cases}
    \end{equation}
    Thus, we use the variable $X_{n+1}$ only when the past statistic $\bar D_n$ is non-negative. 
    \item Statistic calculation: Start with $\bar D_0 = 0$ and calculate $\{\bar D_n\}$ as follows: 
    \begin{equation}
    \label{modified D_n_theta_star_1}
        \bar D_{n+1} = \begin{cases}
            \left( \bar D_n + \log\frac{\bar g(X_{n+1})}{f(X_{n+1})}\right)^{h+}, \quad \text{if } \bar M_{n+1} = 1, \\
         \min\{\bar D_{n} + \mu, 0\}, \quad \quad \quad \quad \text{if } \bar  M_{n+1} = 0. 
        \end{cases}
    \end{equation}
    When $X_{n+1}$ is used for decision making ($\bar M_{n+1}=1$), we update the test statistic using the usual CUSUM update, except we allow it to go below zero through the operation $(x)^{h+}:=\max\{x, -h\}$. If the statistic $\bar D_{n+1}$ goes below zero because of this operation, the next sample will be skipped.  
    When a sample is skipped, the statistic is updated by adding a carefully chosen constant $\mu$ to the current value of the statistic. 
    \item Stopping rule: Stop at 
\begin{align}\label{tau M_n}
    \bar\tau_{\text{rdc}} = \inf\{n \geq 1 : \bar D_n \geq A\}.
\end{align}
Thus, we stop the first time the statistic $\bar D_n$ crosses a carefully designed threshold $A$. 
\end{enumerate}
\end{algorithm}

Overall, the algorithm works as follows. We start with $\bar D_0 = 0$. Thus, the first sample is always taken $M_1=1$ and the statistic is updated using  $\bar D_{1} =
            \max\{ \bar D_0 + \log[\bar g(X_{1})]/f(X_{1})], {-h} \}$.
            As long as the statistic $\bar D_n \geq 0$, this update is repeated for every collected sample. Once $\bar D_n$ goes below zero, a number of consecutive samples are skipped based on the undershoot of the statistic. This undershoot is truncated at $-h$ for mathematical convenience (although $h$ can be any arbitrary but fixed value). Since the undershoot distribution depends on the given pre- and post-change laws, it cannot be directly controlled. Thus, to control the average number of observations used, a parameter $\mu > 0$ is used to skip $\lceil \textit{undershoot}/\mu \rceil$ number of observations, where $\lceil x\rceil$ represents the smallest integer bigger than $x$. 
           Thus, the smaller the value of $\mu$, the larger the number of consecutive samples skipped for any fixed undershoot value. As with most standard change detection algorithms, a change is declared when the statistic is above some threshold $A$—a large value of $A$ results in a lower value of the rate of false alarms. We will later show that the constraints on $\far$ and $\pdc$ can be met approximately independently of each other.

We now compare the structure of the RDE-CUSUM algorithm with other standard algorithms in the literature. When $\mu=0, h=0$, the RDE-CUSUM algorithm reduces to the robust CUSUM algorithm introduced in \cite{unni-etal-ieeeit-2011}. When the post-change law is precisely known, the post-change family is a singleton, i.e., $\mathcal{G} = \{\bar{g}\}$, and $\mu>0, h>0$, then the RDE-CUSUM algorithm reduces to the DE-CUSUM algorithm introduced in \cite{banerjee2013data}. Finally, when 
$\mathcal{G} = \{\bar{g}\}$ and $\mu=0, h=0$, the RDE-CUSUM algorithm reduces to the classical CUSUM algorithm. 

\section{Asymptotic Optimality of RDE-CUSUM Algorithm}
\label{sec:Optimality}
To prove the asymptotic optimality of the RDE-CUSUM algorithm, we first establish a universal lower bound on the performance of any policy. 

\begin{theorem}[Universal Lower Bound]
\label{thm:LowerBound}
    Let $\Psi$ be any policy for QCD with sampling control that satisfies the $\far(\Psi) \leq \alpha$ and $\pdc(\Psi) \leq \beta$, for any fixed $0 < \alpha < 1$ and $0 < \beta < 1$. Then, 
    \begin{equation}
        \sup_{G \in \mg} \  \wadd^{G}(\Psi) \geq \frac{|\log \alpha|}{\dd_{\textnormal{KL}}(\bar g \; \| \; f)}(1+o(1)),
    \end{equation}
    where the $o(1)$ term goes to zero as $\alpha \to 0$. 
\end{theorem}
\begin{proof}
    Note that
    $$
    \sup_{G \in \mg} \  \wadd^{G}(\Psi) \geq \wadd^{\bar G}(\Psi)
    $$
    When the pre-change law is $f$ and the post-change law is $\bar g$, it is well-known (\cite{lai-ieeetit-1998, lord-amstat-1971}) that 
    \begin{equation}
         \wadd^{\bar G}(\Psi) \geq \frac{|\log \alpha|}{\dd_{\textnormal{KL}}(\bar g \; \| \; f)}(1+o(1)).
    \end{equation}
\end{proof}
Note that for any fixed $\beta$, the lower bound only depends on the $\far$ constraint $\alpha$ and does not depend on $\beta$. 

To show that the proposed algorithm RDE-CUSUM achieves the lower bound, we first prove a theorem on its $\far$ and $\pdc$ performances. 

\begin{theorem}\label{thm:FARPDC}
The \textup{RDE-CUSUM} algorithm $\bar{\Psi}_{\text{RDC}}$ satisfies the following $\far$ and $\pdc$ bounds.  
    \begin{enumerate}
        \item False Alarm Rate constraint $\alpha$: For any fixed $h$ and $\mu$, setting $A=|\log \alpha|$ guarantees that
        $$
        \far(\bar{\Psi}_{\text{RDC}}) \leq \alpha. 
        $$
        Thus, the $\far$ constraint can be met irrespective of the choices for $h$ and $\mu$ to satisfy the $\pdc$ constraint $\beta$. 
        \item Pre-change duty cycle constraint $\beta$: For any choice of the threshold $A > 0$ and parameter $h>0$, there are constants $C_1$ and $C_2(h)$ (that are not a function of $A$) 
 such that setting 
        $$
        \mu \leq \frac{\beta}{1-\beta} \frac{C_2(h)}{C_1} 
        $$
        guarantees that
        $$
        \pdc(\bar{\Psi}_{\text{RDC}}) \leq \beta. 
        $$
        Here we use $C_2(h)$ to denote that the constant $C_2$ depends on $h$. 
        Thus, the $\pdc$ constraint can be met irrespective of the choice of threshold $A$ chosen to satisfy the $\far$ constraint $\alpha$.  
    \end{enumerate}
\end{theorem}
\begin{proof}
See Appendix \ref{sec:AppendixFAR}.
\end{proof}
We remark that the $\far$ bound is universal and does not depend on the LFL $\bar g$. However, the choice of $\mu$ to satisfy the $\pdc$ constraint depends on $\bar g$ through the constants $C_1$ and $C_2$ appearing in the theorem. 

When $A \to \infty$ and $h \to \infty$, the choice of $\mu$ is considerably simplified. 
\begin{corollary}[Simple Choice of $\mu$]
\label{corr:simplePDC}
   If $A \to \infty$ and $h \to \infty$, then by selecting $\mu$ such that 
    \begin{align}\label{eq:mu choice}
    \mu \leq \frac{\beta}{1 - \beta} \dd_{\textnormal{KL}}(f \; \| \; \bar g),
    \end{align}
    the $\pdc$ constraint of $\beta$ can be satisfied in that asymptotic regime, or asymptotically satisfied.  
\end{corollary}
\begin{proof}
    See Appendix \ref{sec:AppendixFAR}.
\end{proof}

\medskip
Next, we provide a delay analysis of the RDE-CUSUM algorithm and obtain an upper bound on $\sup_{G \in \mg} \  \wadd^{G}(\bar{\Psi}_{\text{RDC}})$ for any choice of threshold $A$. 
\begin{theorem}
\label{thm:DelayAnalysis}
    If Assumption~\ref{assumption1} is satisfied, then for any fixed choices for $A, \mu, h$, the delay of the \textup{RDE-CUSUM} algorithm satisfies
    \begin{equation}
        \sup_{G \in \mg} \  \wadd^{G}(\bar{\Psi}_{\text{RDC}}) \leq \frac{A}{\dd_{\textnormal{KL}}(\bar g \; \| \; f)}(1+o(1)),
    \end{equation}
  where the $o(1)$ term goes to zero as $A \to \infty$.
\end{theorem}
\begin{proof}
    See Appendix~\ref{appendix 1}.
\end{proof}

Based on the results above, we can now state and prove the result on the asymptotic robust optimality of the RDE-CUSUM algorithm. 

\begin{theorem}[Asymptotic robust optimality of RDE-CUSUM]
\label{thm:AsymptoticOptimality}
If Assumption~\ref{assumption1} is satisfied, then the \textup{RDE-CUSUM} algorithm is asymptotically robust optimal for each fixed $\beta$, as $\alpha \to 0$.  
\end{theorem}
\begin{proof}
By Theorem~\ref{thm:FARPDC}, the $\far$ constraint $\alpha$ is satisfied if we set the threshold $A= | \log \alpha |$. Fix any $0 < h < \infty$. By setting $
        \mu \leq \frac{\beta}{1-\beta} \frac{C_2(h)}{C_1} 
        $, the $\pdc$ constraint is satisfied. Thus, both $\far$ and $\pdc$ constraint can be met and for $A= | \log \alpha |$, Theorem~\ref{thm:DelayAnalysis} shows that
        \begin{equation}
        \sup_{G \in \mg} \  \wadd^{G}(\bar{\Psi}_{\text{RDC}}) \leq \frac{A}{\dd_{\textnormal{KL}}(\bar g \; \| \; f)}(1+o(1)) = \frac{| \log \alpha |}{\dd_{\textnormal{KL}}(\bar g \; \| \; f)}(1+o(1)).
    \end{equation}
    By Theorem~\ref{thm:LowerBound}, the obtained upper bound is also the lower bound. This proves the asymptotic robust optimality of the RDE-CUSUM algorithm. 
\end{proof}



\section{Examples of Least Favorable Law}
\label{sec:Examples}
In this section, we provide examples of LFL from Gaussian and Poisson families. For more similar and more general examples, we refer the readers to \cite{unni-etal-ieeeit-2011, bane-hou-multinonstationary, hou2024robust}. 


\begin{example}[Gaussian LFL]
\label{exam:GaussLFD}
    Let the pre-change density be given by $f = \mathcal{N}(0,1)$ and the post-change densities be given by
    $$
    g = \mathcal{N}(\mu, 1), \quad \mu \geq \bar\mu,
    $$
    where the mean is not known but is believed to be lower bounded by some known $\bar{\mu}$. We show that $\bar g = \mathcal{N}(\bar \mu, 1)$ is LFL.    
    Since
    \begin{align*}
        \log\frac{\bar{g}(X)}{f(X)} = \bar{\mu} X - \frac{\bar{\mu}^2}{2},
    \end{align*}
     we have 
     \begin{align*}
        X \sim  \mathcal{N}(\bar{\mu}, 1) &\implies \log \frac{\bar{g}_(X)}{f(X)}  \sim \mathcal{N} \left( \frac{\bar{\mu}^2}{2}, \ \bar{\mu}^2\right), \\
         X \sim  \mathcal{N}({\mu}, 1) &\implies \log \frac{\bar{g}(X)}{f(X)}  \sim \mathcal{N} \left(\bar{\mu}\cdot{\mu} - \frac{\bar{\mu}^2}{2}, \; \bar{\mu}^2\right).
     \end{align*}
      Since $\mu \geq \bar{\mu}$, we have
      $
      \bar{\mu}\cdot{\mu} - \frac{\bar{\mu}^2}{2} \geq \frac{\bar{\mu}^2}{2}.
      $
      Thus, $\mathcal{N}\left(\bar{\mu}\cdot{\mu} - \frac{\bar{\mu}^2}{2}, \; \bar{\mu}^2\right)$ dominates $\mathcal{N}\left( \frac{\bar{\mu}^2}{2}, \; \bar{\mu}^2\right)$ in stochastic order \cite{bane-hou-multinonstationary, hou2024robust}.  
      
\end{example}

\medskip
\begin{example}[Poisson LFL]
\label{exam:PoissonLFD}
    Let the pre-change density be given by
    $
    f = \text{Pois}(\lambda_0)
    $
    and the post-change densities be given by
    $$
    g = \text{Pois}(\lambda_1), \quad \lambda_1 \geq \bar\lambda_1,
    $$
    where the mean is not known but is believed to be lower bounded by some known $\bar{\lambda}$. 
    We show that $\bar g = \text{Pois}(\bar \lambda_1)$ is LFL. 
    Since
    \begin{align*}
        \log \frac{\bar{g}(X)}{f(X)} = \log \left[\left(\frac{\bar{\lambda}_{1}}{\lambda_0}\right)^X e^{-\bar{\lambda}_{1} + \lambda_0}\right] = X \log \left(\frac{\bar{\lambda}_{1}}{\lambda_0}\right) -\bar{\lambda}_{1} + \lambda_0,
    \end{align*}
     we have 
     \begin{align*}
        X \sim  \text{Pois}(\bar{\lambda}_{1}) &\implies \log \frac{\bar{g}(X)}{f(X)}  \sim \text{Pois} \left(\bar{\lambda}_{1} \log \left(\frac{\bar{\lambda}_{1}}{\lambda_0}\right) -\bar{\lambda}_{1} + \lambda_0\right), \\
         X \sim  \text{Pois}({\lambda}_{1}) &\implies \log \frac{\bar{g}(X)}{f(X)}  \sim \text{Pois} \left({\lambda}_{1} \log \left(\frac{\bar{\lambda}_{1}}{\lambda_0}\right) -\bar{\lambda}_{1} + \lambda_0\right) .
     \end{align*}
    Since $\lambda_{1} \geq \bar{\lambda}_{1}$, we have
    $$
    {\lambda}_{1} \log \left(\frac{\bar{\lambda}_{1}}{\lambda_0}\right) -\bar{\lambda}_{1} + \lambda_0 \geq \bar{\lambda}_{1} \log \left(\frac{\bar{\lambda}_{1}}{\lambda_0}\right) -\bar{\lambda}_{1} + \lambda_0.
    $$
    Thus, $\text{Pois} \left({\lambda}_{1} \log \left(\frac{\bar{\lambda}_{1}}{\lambda_0}\right) -\bar{\lambda}_{1} + \lambda_0\right)$ stochastically dominates $\text{Pois} \left(\bar{\lambda}_{1} \log \left(\frac{\bar{\lambda}_{1}}{\lambda_0}\right) -\bar{\lambda}_{1} + \lambda_0\right)$ \cite{bane-hou-multinonstationary, hou2024robust}. 
    
      
\end{example}

\section{Numerical Studies}\label{sec: numerical studies}
\label{sec:Numerical}
We demonstrate the effectiveness of the RDE-CUSUM algorithm on both simulated and real data. 
First, we compare the detection time of robust algorithms across different values of $\beta$ in the $\pdc$ constraint, using simulated continuous and discrete data processes, Gaussian and Poisson. Next, we validate its effectiveness through extensive simulations. We also apply the RDE-CUSUM algorithm to detect the onset of a pandemic. The results indicate that the RDE-CUSUM algorithm, which adaptively skips observations, performs nearly as effectively as the robust CUSUM test, which does not skip observations.

\subsection{RDE-CUSUM Algorithm on Simulated Data with Unknown Density}
We first consider a Gaussian example: 
\begin{equation}\label{eq: real_data_normal} 
f = \mathcal{N}(0,1), \quad g = \mathcal{N}(\theta, 1), \quad \theta \geq 0.5. 
\end{equation} 
As discussed in Section~\ref{sec:Examples}, the LFL is 
$\bar{g} = \mathcal{N}(0.5, 1).$
In Figure \ref{fig:comparison 1} (right), we compare the detection delays ($\wadd$) of three schemes:
\begin{enumerate}
    \item the RDE-CUSUM algorithm, 
    \item the robust CUSUM algorithm (the RDE-CUSUM algorithm with no observation control), and 
    \item the fractional sampling scheme, where the robust CUSUM statistic is updated based on a fair coin toss: if a head is obtained, the new observation is incorporated to update the robust CUSUM statistic; otherwise, the previous statistic value is retained. The first observation is always included.
\end{enumerate}
For the RDE-CUSUM algorithm \eqref{tau M_n}, $h = 10$ and $\mu$ is chosen according to \eqref{eq:mu choice}.
\begin{figure}[!h]
    \centering
    \includegraphics[scale=0.6]{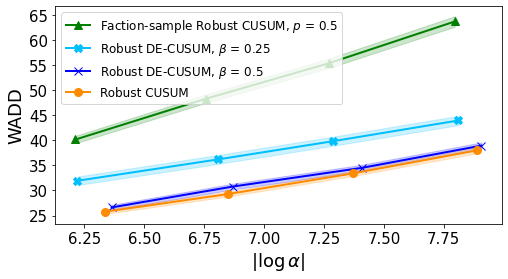}
    \caption{Comparison of RDE-CUSUM algorithm, robust CUSUM algorithm, and the fractional sampling scheme-based robust CUSUM algorithm for Gaussian data processes \eqref{eq: real_data_normal} with $\theta = 1$.}
    \label{fig:comparison 1}
\end{figure}
We compare the $\wadd$ across different false alarm rates. Each point represents an empirical average from 5000 simulations, with carefully chosen thresholds to ensure comparable false alarm rates. The shaded area represents the $95\%$ confidence interval. When $\beta = 0.5$, the detection delay of the RDE-CUSUM algorithm closely matches that of the robust CUSUM, the robust optimal solution with known distribution that uses all the observations. For smaller $\beta$ values, $\wadd$ increases, but it remains lower than that of the fraction sample algorithm.
The plot shows that one can skip $50\%$ of the observations used in the detection process without significantly affecting the detection delay and without knowing the post-change distribution.

In Figure \ref{fig:comparison 2}, we report a similar result for Poisson data: 
\begin{equation}\label{eq: real_data_poisson}
    \begin{split}
        f = \text{Pois}(0.5), \quad  g = \text{Pois}(\lambda), \quad \lambda \geq 1,
    \end{split}
\end{equation}
for which $\bar g = \text{Pois}(1)$. 
\begin{figure}[!h]
    \centering
    \includegraphics[scale=0.6]{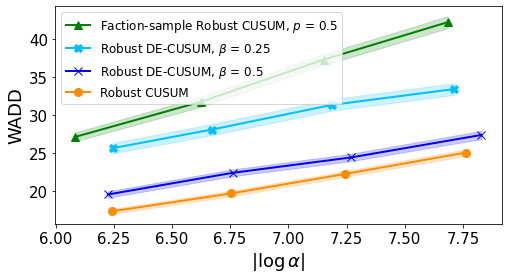}
    \caption{Comparison of RDE-CUSUM algorithm, robust CUSUM algorithm, and the fractional sampling scheme-based robust CUSUM algorithm for Gaussian data processes \eqref{eq: real_data_poisson} with $\lambda = 1.5$.}
    \label{fig:comparison 2}
\end{figure}



\subsection{Robust DE-CUSUM Test on Detecting Pandemic Onset}
\label{sec: DE-CUSUM on covid}
In this section, we apply the RDE-CUSUM algorithm, the robust CUSUM, and the fractional sampling scheme-based robust CUSUM algorithm to detect the onset of a pandemic using publicly available U.S. COVID-19 infection data. 

In Figure~\ref{fig:comparison 3} (Left) and Figure~\ref{fig:comparison 4} (Left), we selected two U.S. counties with similar population sizes, Allegheny in PA and St. Louis in MO. The daily case counts for the first 200 days (starting from 2020/1/22) are shown, with $\text{Pois}(1)$ noise added to the data. 
This data generation process emulates scenarios requiring the detection of a pandemic onset amidst daily infections caused by other viruses or the emergence of a new variant. Therefore, the data represents the detection of deviations from an established baseline.

\begin{figure}[!h]
    \centering
    \includegraphics[width=0.475\linewidth]{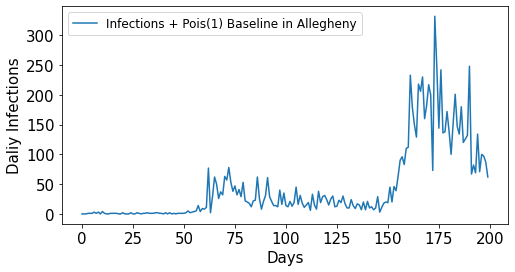}
    \includegraphics[width=0.47\linewidth]{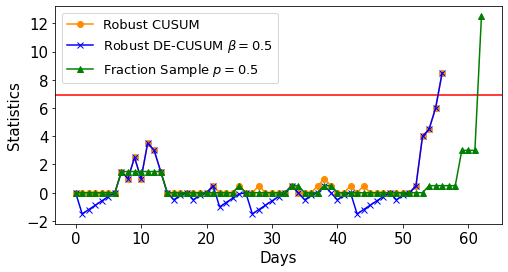}
    \caption{Robust test for detecting COVID-19 outbreak in Allegheny County. [Left] Daily increase of confirmed infection cases with $\text{Pois}(1)$ noise added.  [Right] To detect the pandemic outbreak, robust CUSUM test is designed with pre-change density $f = \text{Pois}(1)$ and post-LFL $\bar g = \text{Pois}(2)$. The threshold is chosen to be $6.9 = \log(1000)$ to guarantee an expected time to false alarm greater than $1000$ days.}
    \label{fig:comparison 3}
\end{figure}

\begin{figure}[!h]
    \centering
    \includegraphics[width=0.475\linewidth]{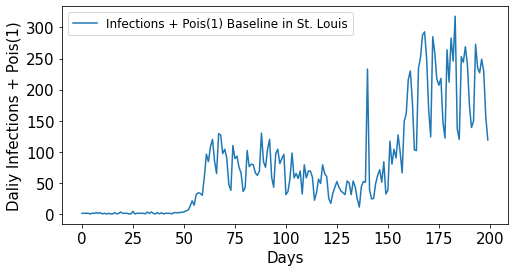}
    \includegraphics[width=0.47\linewidth]{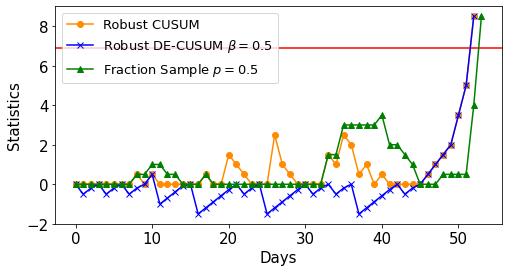}
    \caption{Robust test for detecting COVID-19 outbreak in St. Louis County. [Left] Daily increase of confirmed infection cases with $\text{Pois}(1)$ noise added.  [Right] To detect the pandemic outbreak, robust CUSUM test is designed with pre-change density $f = \text{Pois}(1)$ and post-LFL $\bar g = \text{Pois}(2)$. The threshold is chosen to be $6.9 = \log(1000)$ to guarantee an expected time to false alarm greater than $1000$ days.}
    \label{fig:comparison 4}
\end{figure}

In Figure~\ref{fig:comparison 3} (Right) and Figure~\ref{fig:comparison 4} (Right), since the actual number of infections is unknown, we take a robust approach, designing the CUSUM test with a post-change LFL given by $\text{Pois}(2)$. The threshold is set to $6.9 = \log(1000)$ to ensure an expected time to false alarm exceeding 1000. 
Both the robust CUSUM and robust DE-CUSUM algorithm (with $h = 10$) detect the change promptly (within a week) after daily infections start to rise significantly (a few days past the 50th day in both counties). Given its lower observation cost, the data-efficient algorithm may be preferable for such public health monitoring.
The fractional sampling scheme also detects the outbreak and is notably responsive in St. Louis County. However, it does not consistently match the timely detection of the RDE-CUSUM algorithm, aligning with the results in Figure \ref{fig:comparison 2} (right).

\section{Conclusion}
We proposed the RDE-CUSUM algorithm that uses the likelihood ratio of the data to detect the change and execute an on-off observation control. Since the post-change law is not precisely known, the likelihood ratios are computed using the LFL. A change is declared when the accumulated likelihood ratios cross a threshold. When the accumulated likelihood ratio goes below zero, consecutive samples are skipped based on the undershoot and using a design parameter. After the skipped samples, the accumulation of the likelihood ratios starts again by ignoring the past observations. We showed that the RDE-CUSUM algorithm can be designed to meet the constraints on the false alarm rate and the pre-change duty cycle independently. The algorithm is also asymptotically robust optimal for every fixed constraint on the pre-change duty cycle as the rate of false alarms goes to zero. We showed through simulations that the RDE-CUSUM algorithm can skip a fixed but large fraction of samples without significantly affecting the detection delay and performs comparatively with the robust CUSUM algorithm that uses all the observations. Since the RDE-CUSUM algorithm adaptively skips samples, it performs better than a fractional sampling scheme where observations are skipped based on a coin toss. Finally, we applied the algorithm to detect the onset of the COVID-19 pandemic without using a significant number of samples before the onset.


\appendix
\section{Proof of Theorem~\ref{thm:FARPDC} and Corollary~\ref{corr:simplePDC}}
\label{sec:AppendixFAR}
It follows from the arguments given in the proof of Lemma 4 in \cite{banerjee2013data} that
\begin{equation}
\label{eq:FARPDFproof_1}
    \far(\bar{\Psi}_{\text{RDC}}(A, \mu, h)) \leq \far(\bar{\Psi}_{\text{RDC}}(A, 0, 0)),
\end{equation}
where $\bar{\Psi}_{\text{RDC}}(A, 0, 0)$ is the RDE-CUSUM algorithm with parameters $\mu=h=0$. As discussed earlier, this corresponds to the robust CUSUM algorithm from \cite{unni-etal-ieeeit-2011}. Since robust CUSUM algorithm is the classical CUSUM algorithm with pre-change $f$ and post-change $\bar g$, it follows from the classical analysis of the CUSUM algorithm (\cite{lord-amstat-1971}) that $A=|\log \alpha |$ guarantees that
$$
\far(\bar{\Psi}_{\text{RDC}}(A, 0, 0)) \leq \alpha. 
$$
We refer the readers to \cite{banerjee2013data} to understand why \eqref{eq:FARPDFproof_1} is true. 

To prove the result on the PDC, we need some definitions. We first 
define the stopping time for a sequential probability ratio test using densities $f$ and $\bar g$:
\begin{align*}
    \lambda_A \coloneqq \inf\left\{n \geq 1 : \sum_{i=1}^n \log\frac{\bar g(X_i)}{f(X_i)} \notin [0, A)\right\}.
\end{align*}
It follows from the definitions that 
$$
\bar D_n = \sum_{i=1}^n \log\frac{\bar g(X_i)}{f(X_i)}, \text{ for } n < \lambda_A. 
$$
Thus,
\begin{align*}
    \lambda_A = \inf\left\{n \geq 1 : \bar D_n \notin [0, A)\right\}.
\end{align*}
When $A=\infty$, we get
\begin{align*}
    \lambda_\infty \coloneqq \inf\left\{n \geq 1 : \sum_{i=1}^n \log\frac{\bar g(X_i)}{f(X_i)} < 0\right\}.
\end{align*}
Since $0 < \dd_{\textnormal{KL}}(f \; \| \; \bar g) < \infty$, it follows from \cite{woodroofe1982nonlinear} (see Corollary 2.4) that 
\begin{align*}
    \Expect_{\infty}[\lambda_{\infty}] < \infty,
\end{align*}
Next, to capture the sojourn time below 0, we define for $x < 0$,
\begin{align*}
    T(x) \coloneqq \left\lceil  \frac{\big|(x)^{h+}\big|}{\mu} \right\rceil.
\end{align*}

 It follows from Theorem 5.1 in \cite{banerjee2013data} that for fixed values of $A$, $h$, and $\mu > 0$, if $0 < \dd_{\textnormal{KL}}(f \; \| \; \bar g) < \infty$, then
\begin{align}
\label{eq:FARPDFproof_2}
\pdc(\bar{\Psi}_{\text{RDC}}(A, \mu, h)) = 
\frac{\Expect_{\infty}[\lambda_A | \bar D_{\lambda_A} < 0]}{\Expect_{\infty}[\lambda_A | \bar D_{\lambda_A} < 0] + \Expect_{\infty}[T(\bar D_{\lambda_A}) | \bar D_{\lambda_A} < 0]}.
\end{align}
Define 
$$
\bar Z(X) = \log(\bar g(X)/f(X)). 
$$
Then, it follows from Lemma 1 in \cite{banerjee2013data} that 
\begin{equation}
    \label{eq:FARPDFproof_3}
        \Expect_{\infty}[\lambda_A | \bar D_{\lambda_A} < 0] \leq \frac{\Expect_\infty[\lambda_\infty]}{\Prob_\infty(\bar Z (X_1) < 0)}. 
\end{equation}
Also, it follows from Lemma 2 in \cite{banerjee2013data} that 
\begin{equation}
    \label{eq:FARPDFproof_32}
       \Expect_{\infty}[T(\bar D_{\lambda_A}) | \bar D_{\lambda_A} < 0] \geq  \frac{\Expect_{\infty}\left[\left| \bar Z(X_1)^{h+}\right| \Big |  \bar Z(X_1) < 0 \right]}{\mu} \Prob_\infty(\bar Z (X_1) < 0). 
\end{equation}
Note that these bounds do not depend on the false alarm threshold $A$. Using these bounds, we obtain an upper bound on the $\pdc$ that is not a function of the threshold $A$. Specifically, we have 
\begin{equation}
\label{eq:FARPDFproof_5}
    \begin{split}
        \pdc(\bar{\Psi}_{\text{RDC}}(A, \mu, h)) &= 
\frac{\Expect_{\infty}[\lambda_A | \bar D_{\lambda_A} < 0]}{\Expect_{\infty}[\lambda_A | \bar D_{\lambda_A} < 0] + \Expect_{\infty}[T(\bar D_{\lambda_A}) | \bar D_{\lambda_A} < 0]} \\
&\leq \frac{\frac{\Expect_\infty[\lambda_\infty]}{\Prob_\infty(\bar Z (X_1) < 0)}}{\frac{\Expect_\infty[\lambda_\infty]}{\Prob_\infty(\bar Z (X_1) < 0)} + \frac{\Expect_{\infty}\left[\left| \bar Z(X_1)^{h+}\right| \Big |  \bar Z(X_1) < 0 \right]}{\mu} \Prob_\infty(\bar Z (X_1) < 0)} \\
&= \frac{\Expect_\infty[\lambda_\infty]}{\Expect_\infty[\lambda_\infty] + \frac{1}{\mu}\Expect_{\infty}\left[\left| \bar Z(X_1)^{h+}\right| \Big |  \bar Z(X_1) < 0 \right]\Prob_\infty(\bar Z (X_1) < 0)^2}. 
    \end{split}
\end{equation}
Thus, to satisfy 
\begin{align*}
    \pdc(\bar \Psi_{\textnormal{RDC}}(A, \mu, h)) \leq \beta,
\end{align*}
we must have a $\mu$ that satisfies
\begin{align*}
    \mu \leq \frac{\beta}{1 - \beta} \cdot \frac{\Expect_{\infty}\left[\left| \bar Z(X_1)^{h+}\right| \mid| \bar Z(X_1) < 0 \right] \cdot \Prob_{\infty}\left(\bar Z(X_1) < 0 \right)^2}{\Expect_{\infty}[\lambda_{\infty}]}.
\end{align*}
To prove the theorem, it is enough to set 
\begin{equation}
    \label{eq:FARPDFproof_4}
    C_1 := \Expect_\infty[\lambda_\infty], \quad \quad C_2:=\Expect_{\infty}\left[\left| \bar Z(X_1)^{h+}\right| \mid| \bar Z(X_1) < 0 \right] \cdot \Prob_{\infty}\left(\bar Z(X_1) < 0 \right)^2. 
\end{equation}

To prove the corollary, we take $A \to \infty$ and $h \to \infty$ in the expression of the $\pdc$ to get
\begin{align}\label{eq: PDC approx}
      \pdc(\bar{\Psi}_{\text{RDC}}(A, \mu, h)) &= 
\frac{\Expect_{\infty}[\lambda_A | \bar D_{\lambda_A} < 0]}{\Expect_{\infty}[\lambda_A | \bar D_{\lambda_A} < 0] + \Expect_{\infty}[T(\bar D_{\lambda_A}) | \bar D_{\lambda_A} < 0]}  \\
&\rightarrow \frac{\Expect_{\infty}[\lambda_{\infty}]}{\Expect_{\infty}[\lambda_{\infty}] + \Expect_{\infty}\left[  \lceil | \bar D_{\lambda_{\infty}} |/\mu \rceil \right]} \\
&\leq \frac{\Expect_{\infty}[\lambda_{\infty}]}{\Expect_{\infty}[\lambda_{\infty}] + \Expect_{\infty}\left[  | \bar D_{\lambda_{\infty}} |/\mu \right]} \\
&= \frac{\mu}{\mu + \dd_{\textnormal{KL}}(f \; \| \; \bar g)}.
\end{align}
Here, the last equality follows from Wald's lemma \cite{woodroofe1982nonlinear} which gives
$$
\Expect_{\infty}\left[|\bar D_{\lambda_{\infty}}|\right] = \Expect[\lambda_\infty]\; \dd_{\textnormal{KL}}(f \; \| \; \bar g). 
$$ 
This completes the proof of the corollary. 



\section{Proof of Theorem~\ref{thm:DelayAnalysis}}
\label{appendix 1}
We first obtain a bound on 
\begin{align}
\label{eq:delayappendix1}
    \wadd^{G}(\bar\tau_{rdc}) = \sup_k \esssup \Expect_k^{G}[(\bar\tau_{rdc} - k + 1)^+ | \mathscr{I}_{k-1}]
\end{align}
for any fixed $G \in \mathcal{G}$. For a fixed $k$, the conditional expectation $
\Expect_k^{G}[(\bar\tau_{rdc} - k + 1)^+ | \mathscr{I}_{k-1}]
$
depends on the past history $\mathscr{I}_{k-1}$ only through the value of the statistic $\bar D_{k-1}$ at time $k-1$. As a result, by repeating the arguments similar to those used in Lemma 5 of \cite{banerjee2013data}, it follows that
$$
\Expect_k^{G}[(\bar\tau_{rdc} - k + 1)^+ | \mathscr{I}_{k-1}] \leq \Expect_1^{G}[\bar\tau_{rdc}] + \lceil h/\mu \rceil. 
$$
Since this upper bound no longer depends on the past history or hypothesized change point $k$, we have
\begin{align}
\label{eq:delayappendix2}
    \wadd^{G}(\bar\tau_{rdc}) \leq \Expect_1^{G}[\bar\tau_{rdc}] + \lceil h/\mu \rceil. 
\end{align}

Let $\tau_{rc}$ be the robust CUSUM stopping rule of \cite{unni-etal-ieeeit-2011} and let $\bar W_n$ be the corresponding statistic:
\begin{equation} 
\label{eq:delayappendix3}
    \begin{split}
        \bar W_{n+1} &= (\bar W_n + \log[\bar g(X_{n+1})/f(X_{n+1})])^+, \quad \bar W_0=0,\\
        \tau_{rc} &= \inf\{n \geq 1: \bar W_n \geq A\}.
    \end{split}
\end{equation}
We note that the evolution of RDE-CUSUM statistic $\bar D_n$ and the Robust CUSUM statistic $\bar W_n$ are statistically identical, except of the sojoruns of the statistic $\bar D_n$ below zero. Also, each time the statistic $\bar D_n$ goes below zero, the number of consecutive samples skipped is bounded by $\lceil h/\mu \rceil$. Also, each time $\bar D_n$ comes above zero after being below zero, it restarts at $0$ leading to a new renewal cycle. From here onward, as long as the statistic is above zero, the evolution of $\bar D_n$ is statistically identical to the evolution of $\bar W_n$. Consequently, if  
\begin{align}
\label{eq:delayappendix4}
    \lambda_A = \inf\left\{n \geq 1 : \bar D_n \notin [0, A)\right\},
\end{align}
then the number of times the RDE-CUSUM statistic will go below zero, before hitting the threshold $A$ before hitting zero under the post-change regime, is a geometric random variable with probability 
$$
\Prob_1^G(\bar D_{\lambda_A} \geq A). 
$$
Thus, the mean number of cycles below zero of the RDE-CUSUM statistic is the mean of this geometric random variable:
$$
\frac{1}{\Prob_1^G(\bar D_{\lambda_A} \geq A)}. 
$$
Substituting this in \eqref{eq:delayappendix2} we get
\begin{align}
\label{eq:delayappendix5}
    \wadd^{G}(\bar\tau_{rdc}) &\leq \Expect_1^{G}[\bar\tau_{rdc}] + \lceil h/\mu \rceil \\
    &\leq\Expect_1^{G}[\bar\tau_{rc}] + \frac{\lceil h/\mu\rceil}{\Prob_1^G(\bar D_{\lambda_A} \geq A)} + \lceil h/\mu \rceil. 
\end{align}
Thus,
\begin{equation}
\label{eq:delayappendix6}
    \sup_G \wadd^{G}(\bar\tau_{rdc}) 
    \leq \sup_G \Expect_1^{G}[\bar\tau_{rc}] + \frac{\lceil h/\mu\rceil}{\inf_G \Prob_1^G(\bar D_{\lambda_A} \geq A)} + \lceil h/\mu \rceil. 
\end{equation}

Now, because of Assumption~\ref{assumption1}, we have (see \cite{unni-etal-ieeeit-2011})
\begin{equation}
\label{eq:delayappendix7}
\sup_G \Expect_1^{G}[\bar\tau_{rc}] = \Expect_1^{\bar G}[\bar\tau_{rc}]. 
\end{equation}
Also by the same Assumption~\ref{assumption1}, we have
\begin{equation}
    \label{eq:delayappendix8}
    \begin{split}
    \Prob_1^G(\bar D_{\lambda_A} \geq A) &\stackrel{(a)}{\geq} \Prob_1^G(\bar D_{n} \geq 0, \; \text{ for all }n) \\
    &\stackrel{(b)}{=} \lim_{N \to \infty}\Prob_1^G(\bar D_{n} \geq 0, \; \text{ for all }n \leq N) \\
    &\stackrel{(c)}{=} \lim_{N \to \infty}\Prob_1^G \left(\max_{1 \leq n \leq N} \max_{1 \leq k \leq n} \sum_{i=k}^n \log [\bar g(X_i)/f(X_i)] \geq 0\right) \\
    &\stackrel{(d)}{\geq} \lim_{N \to \infty}\Prob_1^{\bar G} \left(\max_{1 \leq n \leq N} \max_{1 \leq k \leq n} \sum_{i=k}^n \log [\bar g(X_i)/f(X_i)] \geq 0\right) \\
    &=\Prob_1^{\bar G} \left(\max_{1 \leq k\leq n} \sum_{i=k}^n \log [\bar g(X_i)/f(X_i)] \geq 0, \text{ for all }n\right)\\
    &\stackrel{(e)}{>} 0.  
    \end{split}
\end{equation}
In the above equation, inequality $(a)$ is obvious. The equality $(b)$ follows from the continuity of probability. The equality $(c)$ follows from the definition of RDE-CUSUM statistic, which is the same as the robust CUSUM statistic of \cite{unni-etal-ieeeit-2011} when the statistic never goes below zero. The inequality $(d)$ is true because of Lemma III.1 in \cite{unni-etal-ieeeit-2011} as the max and sum operations are monotone increasing in its arguments and because $\bar G$ is the LFL. Finally, the last inequality $(e)$ showing that the probability $\Prob_1^{\bar G} \left(\max_{1 \leq k \leq n} \sum_{i=k}^n \log [\bar g(X_i)/f(X_i)] \geq 0, \text{ for all }n\right)$ is strictly positive follows from the property of a random walk under positive drift \cite{woodroofe1982nonlinear}. Define
\begin{equation}
\label{eq:delayappendix9}
    q := \Prob_1^{\bar G} \left(\max_{1 \leq \leq n} \sum_{i=k}^n \log [\bar g(X_i)/f(X_i)] \geq 0, \text{ for all }n\right) > 0.  
\end{equation}
Then \eqref{eq:delayappendix8} shows that
\begin{equation}
    \label{eq:delayappendix10}
    \inf_{G \in \mathcal{G}}\Prob_1^G(\bar D_{\lambda_A} \geq A) \geq q > 0,
\end{equation}
where $q$ depends on $\bar G$ and not on any other particular $G$. 
Substituting \eqref{eq:delayappendix7}, \eqref{eq:delayappendix8},  \eqref{eq:delayappendix9}, and \eqref{eq:delayappendix10} in \eqref{eq:delayappendix6} we get
\begin{equation}
\label{eq:delayappendix11}
\begin{split}
    \sup_G \wadd^{G}(\bar\tau_{rdc}) 
    &\leq \sup_G \Expect_1^{G}[\bar\tau_{rc}] + \frac{\lceil h/\mu\rceil}{\inf_G \Prob_1^G(\bar D_{\lambda_A} \geq A)} + \lceil h/\mu \rceil \\
    &= \Expect_1^{\bar G}[\bar\tau_{rc}] + \frac{\lceil h/\mu\rceil}{\inf_G \Prob_1^G(\bar D_{\lambda_A} \geq A)}+ \lceil h/\mu \rceil \\
    &\leq \Expect_1^{\bar G}[\bar\tau_{rc}] + \frac{\lceil h/\mu\rceil}{q}+ \lceil h/\mu \rceil \\
    &\leq \frac{A}{\kl(\bar g \ \| \ f)}(1+o(1)) + \frac{\lceil h/\mu\rceil}{q} + \lceil h/\mu \rceil.
    \end{split} 
\end{equation}
Here the last inequality follows from the analysis of the CUSUM algorithm.  This completes the proof of Theorem~\ref{thm:DelayAnalysis}.

\section*{Acknowledgment}
Taposh Banerjee was supported in part by the U.S. National Science Foundation under grant 2427909. Hoda Bidkhori was supported in part by the U.S. National Science Foundation under grant 2427910. 

\bibliographystyle{apalike}
\bibliography{QCD.bib}

\end{document}